\documentclass[11pt,a4paper]{article}
\usepackage{fullpage}
\usepackage{float}
\usepackage[english]{babel}
\usepackage[T1]{fontenc}
\fontencoding{T1}  
\usepackage[utf8]{inputenc}

\usepackage{mathrsfs}
\usepackage{amsfonts}
\usepackage{amssymb}
\usepackage{amsmath}
\usepackage{amsthm}
\usepackage{graphicx}
\usepackage[usenames,dvipsnames]{color}
\usepackage[colorlinks=true,citecolor=blue,linkcolor=magenta]{hyperref}
\usepackage{lmodern}
\usepackage{microtype}
\usepackage{braket}
\usepackage{tabularx,colortbl}
\usepackage{pifont}
\usepackage{mathtools}
\usepackage{color}
\usepackage{dsfont}
\usepackage{algpascal}
\usepackage{algorithm} 
\usepackage{algpseudocode}
\usepackage{chngcntr}
\usepackage{physics}
\usepackage{float} 
\usepackage[table, x11names]{xcolor}
\usepackage{epstopdf}
\usepackage{braket}
\usepackage{caption}   
\usepackage{subfig}   
\usepackage[titletoc,title]{appendix}
\usepackage{todonotes}

\newtheorem{theorem}{Theorem}
\newtheorem{corollary}[theorem]{Corollary}

\newtheorem{proposition}{Proposition}

\newcommand{\comment}[1]{} 
\newcommand{\idty}[1]{\mathbb{1}} 
\newcommand{\ovsqrt}[1]{\frac{1}{\sqrt{2}}}

\newcommand\blfootnote[1]{%
  \begingroup
  \renewcommand\thefootnote{}\footnote{#1}%
  \addtocounter{footnote}{-1}%
  \endgroup
}

\title{Constructing graphs with limited resources}

\begin{document}

\author{Danial Dervovic\thanks{Department of Computer Science, University College London, Gower Street, London WC1E 6BT, United Kingdom.} \and 
Avinash Mocherla\thanks{Imperial College London, Kensington, London SW7 2AZ, United Kingdom.} \footnotemark[1] \and 
Simone Severini\footnotemark[1] \thanks{Shanghai Jiao Tong University, 1954 Huashan Rd, JiaoTong DaXue, Xuhui Qu, Shanghai Shi, China, 200000.}}

\date{\today}

\maketitle

\begin{abstract}
We discuss the amount of physical resources required to construct a given graph, where vertices are added sequentially. We naturally identify information -- distinct into instructions and memory -- and randomness as resources. Not surprisingly, we show that, in this framework, threshold graphs are the simplest possible graphs, since the construction of threshold graphs requires a single bit of instructions for each vertex and no use of memory. Large instructions without memory do not bring any advantage. With one bit of instructions and one bit of memory for each vertex, we can construct a family of perfect graphs that strictly includes threshold graphs. We consider the case in which memory lasts for a single time step, and show that as well as the standard threshold graphs, linear forests are also producible. We show further that the number of random bits (with no memory or instructions) needed to construct any graph is asymptotically the same as required for the Erd\H{o}s-R\'enyi random graph. We also briefly consider constructing trees in this scheme. The problem of defining a hierarchy of graphs in the proposed framework is fully open.
\end{abstract}

\blfootnote{Correspondence address: \href{mailto:d.dervovic@cs.ucl.ac.uk}{\texttt{d.dervovic@cs.ucl.ac.uk}}.}
\section{Introduction}

\emph{Alice} (\textsf{A}) and \emph{Bob} (\textsf{B}) work together to construct a graph. At
each time step $t\in \mathbb{N}$, \textsf{A} sends \emph{instructions} to 
\textsf{B}. Then, \textsf{B} constructs a graph $G_{t}=(V(G_t),E(G_t))$ according to
the instructions, by adding vertices and edges to the graph $G_{t-1}$. We
will try to work in the simplest possible setting. It will be convenient to assume that $G_{t}$ is on $t$ vertices, $1,2,\ldots,t$, and that
vertex $t$ is added at time $t$. Thus, the neighbours of vertex $t$ at time $%
t$ are vertices of $G_{t-1}$. In
addition to instructions, \emph{i.e.}, \emph{information}, \textsf{B} may
have access to a \emph{memory} and a source of \emph{randomness} (\emph{e.g.}%
, \textsf{B} can flip a coin). Of course, memory is also information, but it
seems useful to distinguish between information received at time $t$ from
information received at time $t^{\prime }<t$. These three quantities will be
called \emph{resources}. Notice that we do not consider potential loss of
information. In other words, when we say \textquotedblleft \textsf{A} sends
instructions to \textsf{B}\textquotedblright , we do not mean that there is
an information-theoretic (noisy) channel between \textsf{A }and \textsf{B}.

Instructions and memory consist of bit strings -- more generally, we could
use the symbols of an alphabet. We assume that memory is not required to
remember the effect of each instruction. For instance, if the instruction $%
1011$ at time $t$ is associated to a certain action, say $a$, then $1011$ is
associated to $a$ at every time $t^{\prime}>t$, but there is no need to
store $1011$. Obviously, but importantly, no action can be performed without
the relevant resources -- \emph{e.g.}, adding an edge between a specific
pair of vertices, making a random choice, \emph{etc.} Since $(t^{2}-t)/2$
bits are needed to specify $G_{t}$ exactly, it is intuitive that \textsf{B }%
requires $\Theta(t^{2})$ coin flips to construct $G_{t}$ with randomness
only, whenever instructions are unavailable -- see Section~\ref{sec:randomness} for a thorough description of this point. The process introduced in \cite{JS13}, which was
originally studied in the context of graph limits, aimed at defining a
notion of \emph{likelihood} for a graph $G_{t}$ constructed exactly in the
way described above.

Here, we take a different perspective and look at the amount of instructions,
memory, and randomness needed by \textsf{A} and \textsf{B} for constructing
graphs. We mostly focus on the role of \textsf{B}. The proposed framework
has clearly many variations. We skim over the basic ones. The central
question is: what graphs can \textsf{B} construct with the use of limited
resources? \textsf{B} has full freedom of interpreting instructions, but no
computational power (\emph{e.g.}, \textsf{B} cannot count). A setting in
which graphs are seen as constructed/defined by computationally bounded
distributed agents is proposed by Arora \emph{et al.} \cite{A}.

We will need some standard notation. Given graphs $G$ and $H$, the \emph{%
disjoint union} is the graph $G\uplus H$ such that $V(G\uplus H) = V(G)\uplus V(H)$ and $\{i,j\}\in E(G\uplus H)$ if and only if $\{i,j\}\in
E(G)$ or $\{i,j\}\in E(H)$. We denote by $P_{n}$, $C_{n}$, $K_{n}$, $E_{n}$,
and $K_{l,m}$, the $n$-path, the $n$-cycle, the complete graph on $n$
vertices, the empty graph on $n$ vertices, and the complete bipartite graph
on $l+m$ vertices, respectively. (See \cite{D} for the elementary graph
theory that we use.)

Let us denote by $A(G)$ the adjacency matrix of $G$. The $(i,j)$-th entry of $A(G)$ is\ $[A(G)]_{i,j}=1$ if $\{i,j\}\in E(G)$ and $[A(G)]_{i,j}=0$,
otherwise. The following fact is obvious, but worth a mention for the sake
of completeness:

\begin{proposition}
Every graph can be constructed by adding vertices sequentially.
\end{proposition}

\begin{proof}
Given a graph $G$ on $n$ vertices, (when $G$ is undirected)\ the entries in
the \emph{triangle} above the main diagonal are the significant part of $A(G)
$. Label the vertices of $G$ by $1,2,\ldots,n$, where vertex $1$ corresponds to
the $n$-th line of $A(G)$, vertex $2$ to the $(n-1)$-th line, \emph{etc}. --
a \emph{line} is a row or a column. Let us assume that at time $t$ we add
vertex $t=1,2,\ldots,n$. The proof follows by observing that in $G_{t}$ the
neighbours of vertex $t\in V(G_{t})$ correspond to lines $t-i$, with $%
i=2,\ldots,t-1$, of $A(G)$, and that these lines have been already completed
when adding each vertex $i$.
\end{proof}

\bigskip 

We note that \textit{threshold} graphs occur frequently in this model. It is worth providing a proper description of what such a graph is.  A \emph{threshold graph} is a graph which can be constructed from $K_{1}$ by a sequence of two
operations:\ add an isolated vertex or add a dominating vertex. Recall that a
vertex of a graph $G$ is \emph{isolated} if it has degree zero and \emph{%
dominating} if it is adjacent to all other vertices of $G$. The family of
threshold graphs is denoted by $\mathcal{T}$. We make use of the forbidden subgraphs
characterization of $\mathcal{T}$ when classifying a graph as being threshold. This may also be taken as a definition  
\cite{G80}: a graph $G$ is a \emph{threshold graph} if $G$ does not contain $%
P_{4}$, $C_{4}$ and $K_{2}\uplus K_{2}$ as induced subgraphs. Threshold
graphs are unigraphs (that is, completely specified by their degree sequence
up to isomorphism)\ and easy to recognise. Many hard computational problems
are efficiently solvable when threshold graphs are taken as instances.

In the following sections, we consider the graphs that can be constructed by adding vertices sequentially when \textsf{B} has access to different combinations of resources. First, we will study the simplest case, where \textsf{B} receives one bit instructions from \textsf{A}. With no other information, we show that \textsf{B} is limited to constructing threshold graphs. In fact, these graphs prove to be constructible with any combination of resources. In the next case, \textsf{B} is also given access to one bit of memory, which is the ability to label each vertex with either $0$ or $1$ as it is placed, and thus differentiate the vertex set into two groups. We see that \textsf{A} is now able to give instructions which are directed at only one group of vertices, and this  extends the family of constructible graphs beyond threshold graphs, including non-threshold graphs such as bipartite and split graphs. Memory, as it turns out, has further uses when its use is modified. Specifically, we explore two variations of memory available to \textsf{B}: \textit{fading} and \textit{modifiable}. In the first variation, \textit{fading} memory only lasts one time step, which means that \textsf{B} can only remember the label of the last two placed vertices. This allows the formation of a linear forest of path graphs, among other non-threshold graphs. In the second case - \textit{modifiable memory} - memory can be used to modify the edges between previously placed vertices. Despite offering a new construction mechanism for \textsf{B}, it is shown that this does not extend the number of unique graphs beyond normal memory. 
We then consider randomness as the only resource available to \textsf{B}, and show that the resources required to construct any graph $G_{n}$ and the Erd\H{o}s-R\'{e}nyi random graph are equivalent in the asymptotic limit.
In the penultimate section, we modify the procedure for constructing graphs and construct trees using the resources considered previously.
The final section concludes with open questions.

To start, we consider one bit of instruction and no memory, as described above. The main result of this section is the following statement: 

\begin{proposition}\label{prop2}
\label{threshold}Let \textsf{A} send to \textsf{B} one bit
instructions at each time step $t$. Furthermore, assume that \textsf{B} has no
memory and no randomness. Then, $G_{t}$ is a threshold graph.
\end{proposition}

\begin{proof}
Let us consider the simplest possible case, which is, arguably, one bit of instruction per
vertex. \textsf{B} can not distinguish vertices of $G_{t-1}$, by reading their
labels, since \textsf{B} has no memory. \textsf{B} can not distinguish vertices
of $G_{t-1}$, by picking them at random, since \textsf{B} has no randomness. Hence, \textsf{B}, without further information, is restricted to either placing a dominating or isolated vertex for each instruction. 
We list in a table all unique interpretations of a one bit instruction from A that are compatible with these conditions: 

\begin{equation*}
\begin{tabular}{c|c}
$\left( 0\rightarrow \emptyset \right) ,\left( 1\rightarrow \emptyset
\right) $ & $(0\rightarrow E),\left( 1\rightarrow \emptyset \right) $ \\ 
\hline
& $(0\rightarrow E),\left( 1\rightarrow E\right) $%

\end{tabular}%
\ \ \ \ \ \text{ }.
\end{equation*}%

The notation is straightforward:

\begin{itemize}
\item ``$0\rightarrow E$" or ``$1\rightarrow E$" means ``construct $\{t,i\}\in E(G_{t})$ for every $%
i\in V(G_{t-1})$ when receiving instruction $0$ or $1$ at time $t$", that is, $t$ is a dominating vertex;

\item ``$0\rightarrow \emptyset $" or ``$1\rightarrow \emptyset $" means ``construct $\{t,i\}\notin E(G_{t})$
for every $i\in V(G_{t-1})$ when receiving instruction $0$ or $1$ at time $t$", that is, $t$ is an isolated vertex.

\end{itemize}

The cases $\left( 0\rightarrow \emptyset \right) ,\left(
1\rightarrow \emptyset \right) $ and $(0\rightarrow E),\left( 1\rightarrow
E\right) $ give $G_{t}=E_{t}$ and $G_{t}=K_{t}$, respectively. These are
both threshold graphs. It follows directly from the definition of a
threshold graph that the case $(0\rightarrow E),\left( 1\rightarrow
\emptyset \right) $ gives a threshold graph, since $(0\rightarrow E)$
introduces a dominating vertex and $\left( 1\rightarrow \emptyset \right) $
introduces an isolated vertex. It is clear that instructions of
arbitrary length will not change the above table, because the number of
possible interpretations of the instructions does not increase. In summary, the families of graphs constructed by the given instructions are:
\begin{equation*}
\begin{tabular}{c|c}
$\left( 0\rightarrow \emptyset \right) ,\left( 1\rightarrow \emptyset
\right) $ & $(0\rightarrow E),\left( 1\rightarrow \emptyset \right) $ \\ 
$E_{l}\uplus E_{m}$ & $\mathcal{T}$ \\ \hline
& $(0\rightarrow E),\left( 1\rightarrow E\right) $ \\ 
& $K_{l+m}$%
\end{tabular}%
\ \ \ \ \ \text{ }.
\end{equation*}%
We specify with $l$ and $m$ the number of vertices associated with
instruction $0$ and $1$ respectively.
\end{proof}

\bigskip 

\section{One bit of instruction and one bit of memory}\label{sec:oneb}

Let us now consider families of graphs that \textsf{B} can construct with the use of an extra resource: memory. The construction method is as before. One bit instructions are sent from \textsf{A} to \textsf{B}. For each instruction, \textsf{B} adds a vertex, and then carries out the instruction specified by said vertex. Unlike in the previous case, where all vertices were considered the same, the addition of one bit of memory  allows \textsf{B} to label each vertex with the instruction 0 or 1, as it is placed. Hence \textsf{B} can differentiate between two different groups of vertices. \textsf{A} is now able to send instructions which are directed at one particular group of vertices, and this increases the number of graphs that can be constructed.  For consistency, we assert that no other edge can be added to $G_{t-1}$ at
time $t$, meaning that, at time $t^{\prime }\geq t$, \textsf{B} is not allowed to
add an edge of the form $\{i,j\}$ with both $i$ and $j$ in $G_{t-1}$ (The effects of removing this restriction will be studied in section 3). 
In other words: for all $t' < t$, $G_{t'}$ is an induced subgraph of $G_t$ on the vertices $\{1, \ldots, t'\}$.
 
For the following analysis, we must define the standard join operation. Given graphs $G$ and $H$, their \emph{join}, $G+H$, is the graph constructed by taking the disjoint union of $G$ and $H$, then connecting every vertex of $G$ with every vertex of $H$.
For example, $E_m + E_n = K_{m,n}$ and $K_m + K_n = K_{m+n}$.

Define the labelling function, $\ell: V(G) \to \{0, 1\}$ that returns the label of a given vertex.
Unless stated otherwise, we will assume the bit string of instructions $x$ sent by \textsf{A}, has $l$ vertices labelled `0' and $m$ vertices labelled `1', that is $\abs{\{i \in V(G_t)\ | \ \ell(i) = 0 \}} = l$ and $\abs{\{i \in V(G_t)\ | \ \ell(i) = 1 \}} = m$.

We also define the graphs $E_{(x)}$, $K_{(x)}$ and $\widetilde{K}_{(x)}$.
First, $x = x_1x_2\ldots x_n \in \{0,1\}^n$ is a bitstring of length $n$. 
$E_{(x)}$ is the complete bipartite graph $K_{l,m}$, where the bipartition classes are vertices labelled `0' and `1', minus edges of the form $\{t, t'\}$ when $x_t=0$, $x_{t'} = 1$ and $t < t'$.
Similarly, $K_{(x) }$ is the complete split graph $K_{l}+ E_{m}$  minus edges of the form $\{t, t'\}$ when $x_t=0$, $x_{t'} = 1$ and $t>t'$.
The graph $\widetilde{K}_{(x)}$ is the graph $K_{l} \uplus K_{m}$ plus edges of the form $\{t, t'\}$ when $x_t=0$, $x_{t'} = 1$ and $t < t'$.
Our analysis results in the following proposition:

\begin{proposition}\label{oneb}
Let \textsf{A} send to \textsf{B} one bit of instruction at
each time step $t$. Let us assume that \textsf{B} has no randomness.
Moreover, let us assume that \textsf{B} can assign one bit of memory for each vertex. Then, $G_{t}$ is
either a threshold graph or one of the following graphs (not necessarily
threshold): 
\begin{equation}
E_{l}\uplus E_{m}, \  E_{(x)}, \  K_{l}\uplus E_{m}, \  K_{l,m},K_{(x)}, \  K_{l}+
E_{m}, \  K_{l}\uplus K_{m}, \  \widetilde{K}_{(x)}, \  K_{l+m}.
\label{onebg}
\end{equation}
\end{proposition}

\begin{proof}
\textsf{B} can distinguish vertices of $G_{t}$ by reading their labels, since 
\textsf{B} has one bit of memory for each vertex. However, the amount of
memory is only sufficient to divide $V(G_{t})$ into two classes: 0 or 1. \textsf{B}
cannot distinguish vertices of $G_{t}$, by picking them at random, since 
\textsf{B} has no randomness. As in the proof of Proposition \ref{threshold}%
, we list in a table all unique interpretations of a one bit instruction from A that are
compatible with these conditions: 
\begin{equation*}
\begin{tabular}{ccc|c}
$\left( 0\rightarrow\emptyset\right) ,\left( 1\rightarrow\emptyset\right) $
& \multicolumn{1}{|c}{$(0\rightarrow1),\left( 1\rightarrow\emptyset\right) $}
& \multicolumn{1}{|c|}{$(0\rightarrow0),\left( 1\rightarrow\emptyset \right) 
$} & $(0\rightarrow E),\left( 1\rightarrow\emptyset\right) $ \\ \hline
& \multicolumn{1}{|c}{$(0\rightarrow1),(1\rightarrow0)$} & 
\multicolumn{1}{|c|}{$(0\rightarrow0),(1\rightarrow0)$} & $(0\rightarrow
E),(1\rightarrow0)$ \\ \cline{2-3}\cline{2-4}
&  & \multicolumn{1}{|c|}{$(0\rightarrow0),\left( 1\rightarrow1\right) $} & $%
(0\rightarrow E),\left( 1\rightarrow1\right) $ \\ \cline{3-4}
&  &  & $(0\rightarrow E),\left( 1\rightarrow E\right) $%
\end{tabular}
\ \ \text{ }.
\end{equation*}
Here, the notation is as follows: 

\begin{itemize}
\item ``$0\rightarrow 1$" or ``$1\rightarrow 1$" means ``construct $\{t,i\}\in E(G_{t})$ for every $%
i\in V(G_{t-1})$ if $\ell(i) = 1$ when receiving instruction $0$ or $1$ at time $t$";

\item ``$0\rightarrow 0 $" or ``$1\rightarrow 0 $" means ``construct $\{t,i\}\in E(G_{t})$
for every $i\in V(G_{t-1})$ if $\ell(i) = 0$ when receiving instruction $0$ or $1$ at time $t$".

\end{itemize}

 We proceed with a case by case analysis:

\begin{description}
\item[$\left( 0\rightarrow\emptyset\right),  \left( 1\rightarrow
\emptyset\right) $:] $G_{t}=E_{l}\uplus E_{m}$ and $G_{t}\in\mathcal{T}$.

\item[$(0\rightarrow 0), \left( 1\rightarrow\emptyset\right) $:] $%
G_{t}=K_{l}\uplus E_{m}$ and $G_{t}\in\mathcal{T}$.

\item[$(0\rightarrow E), \left( 1\rightarrow\emptyset\right) $:] we obtain $%
\mathcal{T}$, since this is the case of Proposition \ref{threshold}.

\item[$(0\rightarrow 0), \left( 1\rightarrow 1\right) $:] $G_{t}=K_{l}\uplus
K_{m}$ and $G_{t}\notin\mathcal{T}$ if $l,m\geq2$.

\item[$(0\rightarrow E), \left( 1\rightarrow E\right) $:] $G_{t}=K_{l+m}$ and $%
G_{t}\in\mathcal{T}$

\item[$(0\rightarrow 1), \left( 1\rightarrow\emptyset\right) $:] $G_t = E_{(x)}$. From the instruction set, we have the edges $\{i, j\}$ if and only if $\ell(i)=0$ and $\ell(j)=1$, $j < i$.
This is the definition of $E_{(x)}$.

\item[$(0\rightarrow 1), (1\rightarrow 0)$:] $G_t = K_{l,m}$. Note that $G_t \not\in \mathcal{T}$ for $t \geq 2$.

\item[$(0\rightarrow 0), (1\rightarrow 0)$:] $G_{t} = K_{(x)}$. The instructions enforce that we have edges $\{i, j\}$ if and only if either: $\ell(i) = \ell(j) = 0$ for $i \neq j$, or $\ell(i) = 0$, $\ell(j) = 1$ for $j > i$.
This is readily seen to be $K_{(x)}$. 

\item[$(0\rightarrow E), (1\rightarrow 0)$:] $G_{t} = K_{l} + E_{m}$. The instructions give us all edges of the form $\{i, j\}$ when either: $\ell(i) = \ell(j) = 0$ for $i \neq j$ or $\ell(i) =0$, $\ell(j) = 1$.
These edges realise the complete split graph $K_{l} + E_{m}$.

\item[$(0\rightarrow E), \left( 1\rightarrow 1\right) $:] $G_t = \widetilde{K}_{(x)}$.
From the instruction set, we have the edges $\{i, j\}$ if and only if either: $\ell(i)= \ell(j)$, or $\ell(i) = 0$, $\ell(j) = 1$, $j < i$.
This graph is $\widetilde{K}_{(x)}$.

\end{description}

The table summarises the families: 
\begin{equation*}
\begin{tabular}{ccc|c}
$\left( 0\rightarrow\emptyset\right) \left( 1\rightarrow\emptyset\right) $ & 
\multicolumn{1}{|c}{$(0\rightarrow1)\left( 1\rightarrow\emptyset\right) $} & 
\multicolumn{1}{|c|}{$(0\rightarrow0)\left( 1\rightarrow\emptyset\right) $}
& $(0\rightarrow E)\left( 1\rightarrow\emptyset\right) $ \\ 
$E_{l}\uplus E_{m}$ & \multicolumn{1}{|c}{$E_{(x)}$} & \multicolumn{1}{|c|}{$%
K_{l}\uplus E_{m}$} & $\mathcal{T}$ \\ \hline
& \multicolumn{1}{|c}{$(0\rightarrow1)(1\rightarrow0)$} & 
\multicolumn{1}{|c|}{$(0\rightarrow0)(1\rightarrow0)$} & $(0\rightarrow
E)(1\rightarrow0)$ \\ 
& \multicolumn{1}{|c}{$K_{l,m}$} & \multicolumn{1}{|c|}{$K_{(x)}$} & $%
K_{l}+ E_{m}$ \\ \cline{2-4}
&  & \multicolumn{1}{|c|}{$(0\rightarrow0)\left( 1\rightarrow1\right) $} & $%
(0\rightarrow E)\left( 1\rightarrow1\right) $ \\ 
&  & \multicolumn{1}{|c|}{$K_{l}\uplus K_{m}$} & $\widetilde{K}_{(x)}$ \\ \cline{3-4}
&  &  & $(0\rightarrow E)\left( 1\rightarrow E\right) $ \\ 
&  &  & $K_{l+m}$%
\end{tabular}
\ \ \ \ .
\end{equation*}
\end{proof}

\bigskip

\section{Using memory to modify $G_{t-1}$} 
In earlier sections, \textsf{B} had the restriction that no other edge could be added to $G_{t-1}$ at time $t$ -- meaning that, at time $t^{\prime }\geq t$, \textsf{B} was not allowed to add an edge of the form $\{i,j\}$ with both $i$ and $j$ in $G_{t-1}$. We can say that graphs produced in this way have the property that $G_{t-1}$ is always an induced subgraph of $G_{t}$. In this section, we relax this restriction, meaning that \textsf{B}, using memory, is now able to construct edges within $G_{t-1}$. We will refer to graphs constructed in this manner as being \textit{memory modifiable}. 

To further illustrate this, let us consider the \textit{memory modifiable} graph $M_{00010}$ constructed from the instructions $(0\rightarrow 1)(1\rightarrow \emptyset)$. As \textsf{B} receives the first four bits from \textsf{A}, an empty graph $E_{4}$ is constructed. This corresponds to $M_{t-1}$. As \textsf{B} has memory, the labels $0$ or $1$ of each of the four vertices can be distinguished.
Hence, when \textsf{B} receives the last instruction bit, which corresponds to joining vertex $0_{5}$ to $1$ vertices, \textsf{B} can instead choose to construct the edges $(0_{1},1_{4})$,$(0_{2},1_{4})$,$(0_{3},1_{4})$, $(0_{5},1_{4})$, as all vertices labeled $0$ are indistinguishable without the above restriction.
The resultant graph $M_{t}$ is a complete bipartite graph, $K_{1,4}$. We notice that if vertex $0_{5}$ was removed, then the remaining induced subgraph is $K_{1,3}$. This is clearly not $E_{4}$. We can now give a definition of a memory modifiable graph: A memory modifiable graph, $M_{t}$, is a graph with the property that its induced subgraph constructed by the removal of vertex $t$ is not isomorphic to $M_{t-1}$.
If \textsf{B} had chosen not to use memory to add edges in $M_{t-1}$, then we construct the equivalent graph listed in Proposition~\ref{oneb}, which for $(0\rightarrow 1)(1\rightarrow \emptyset)$ is the closely related $E_{(00010)}$ (see Proposition~\ref{oneb} and Figure~\ref{fig:ME_fig}).  We can conjecture that for each instruction there exists a memory modifiable graph. Given that the construction method is different the natural question to ask is the following: is the family of memory modifiable graphs the same as the graphs listed in Proposition~\ref{oneb}? 

\bigskip

\begin{figure}[hbtp]
    \centering
    \includegraphics[width=0.4\textwidth]{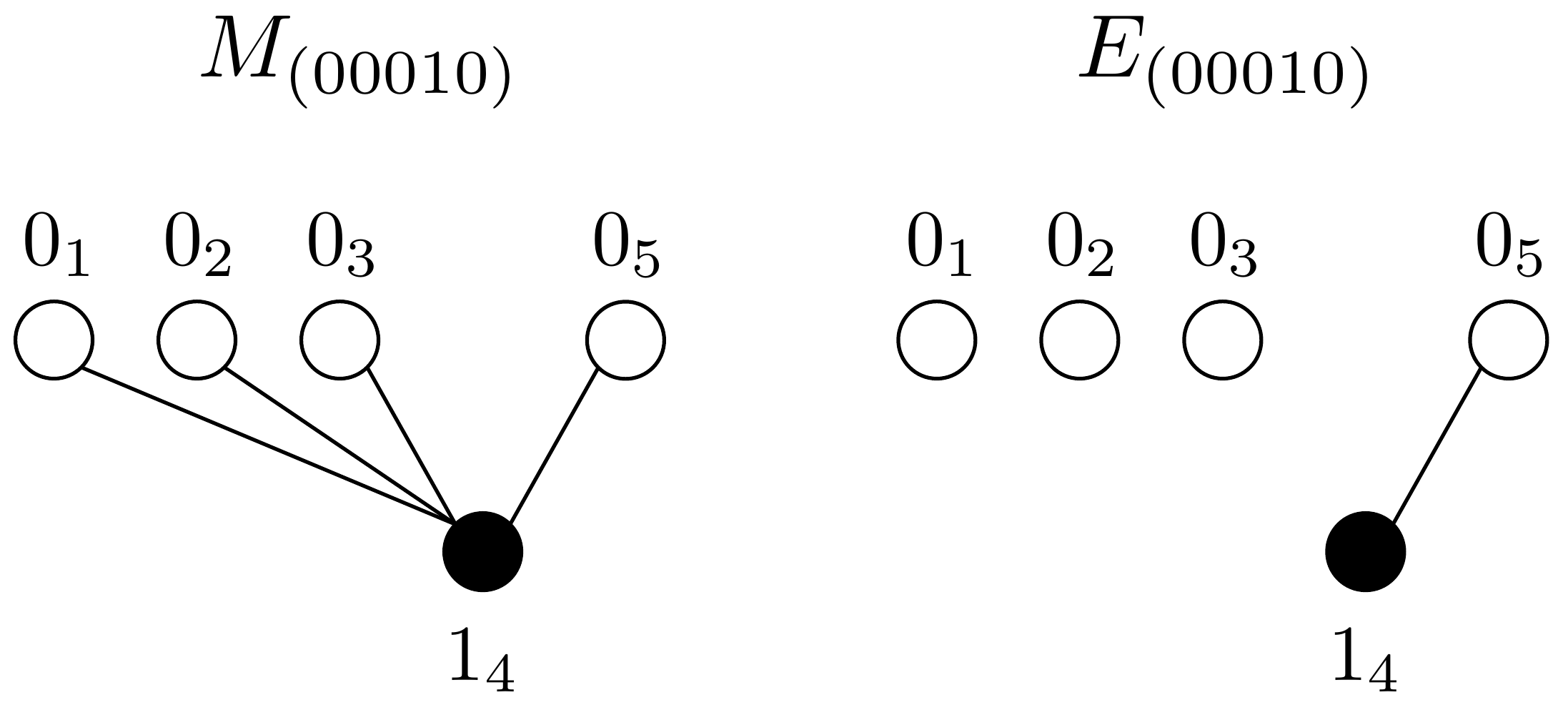}
    \caption{The graphs $M_{(00010)}$ and $E_{(00010)}$, as defined in the main text.}
    \label{fig:ME_fig}
\end{figure}

\begin{corollary}
All possible memory modifiable graphs, $M_{t}$, are of the form: 

\begin{equation}
K_{l} + E_{m} ,\  K_{l,m} ,\ K_{l+m}.
\end{equation}
which are all graphs proven to be constructible in Proposition~\ref{oneb}. 

\end{corollary}

\begin{proof}
Let us assume that $\textsf{B}$ still has access to one bit of memory and instruction per vertex. Let us also remove the restriction that edges cannot be constructed in $G_{t-1}$. Then, from the ten unique instructions identified, only four can produce memory modifiable graphs. For those four, the new graph types are listed. We find that none of these graphs are unique, as they can all be constructed by the method used in Proposition~\ref{oneb}. Below, we assume the following definition for a memory modifiable graph: A memory modifiable graph $M_{t}$ is a graph such that its induced subgraph formed by the removal of vertex $t$ is not isomorphic to $M_{t-1}$. We use this definition to determine whether an instruction can construct memory modifiable graphs. We also note that if a memory modifiable graph has not been formed for a particular instruction set at time  $t$, then the graph $G_{t}$ must be of the form listed in Proposition~\ref{oneb}.

\begin{description}
\item ($0\rightarrow \emptyset)$$(1\rightarrow \emptyset)$: Trivially, this instruction set cannot construct a memory modifiable graph as there are no edges. 

\item ($0\rightarrow E)$$(1\rightarrow E)$:
This instruction always forms a complete graph, $K_{t}$, so cannot form memory modifiable graphs.
  
\item ($0\rightarrow 0)$$(1\rightarrow 1)$: 
Consider $G_{t-1}$ on $t-1$ vertices. This is $K_{l} \uplus K_{m}$. Let us add vertex $t$. If vertex $t$ is 0, then we produce    $K_{l+1} \uplus K_{m}$. If vertex $t$ is 1, then we produce $K_{l} \uplus K_{m+1}$. Now, if vertex $t$ was removed, in either case, we are left with $K_{l} \uplus K_{m}$. Clearly, $G_{t-1}$ is the same as the induced subgraph formed when vertex $t$ is removed.
In fact, we can conclude the same about any set of instructions where the two groups of $0$ and $1$ vertices don't interact, because any graph formed from only one type of vertex, as shown in Proposition~\ref{prop2} can only form $K_{l}$ or $E_{l}$. Neither of these are memory modifiable, as shown above. 

\item ($0\rightarrow E)$$(1\rightarrow 0)$:
Consider $G_{t-1}$ on $t-1$ vertices. This is $K_{l} + E_{m}$. We note that all possible $(1\rightarrow 0)$ edges must exist for this graph type. Let us add vertex $t$. If vertex $t$ is 0, then we add a dominating vertex. If vertex $t$ is 1, then we connect all existing vertices labelled 1 to those labelled 0. However, we see that since all previous $(1\rightarrow 0)$ edges have already been constructed, only vertex $t$ forms connections. In both cases,  if vertex $t$ was removed, then the remaining induced subgraph would be $K_{l} + K_{m}$. Hence, this instruction can never produce memory modifiable graphs. 

\item ($0\rightarrow 0)$$(1\rightarrow \emptyset)$:
This set of instructions is non-interacting. Hence, we can see that this cannot form a memory modifiable graph. 

\item ($0\rightarrow 1)$$(1\rightarrow 0)$:
Consider $G_{t-1}$ on $t-1$ vertices. This is $K_{l,m}$. Let us add vertex $t$. If vertex $t$ is 0, then we construct $K_{l+1,m}$. If vertex $t$ is 1, then we construct $K_{l,m+1}$.
In both cases, vertex $t$ connects to all vertices of the opposite type. However, $t$ doesn't modify the edges of any other pair of vertices, as these have already been connected with previous instructions. If vertex $t$ was removed, in either case, we are left with the original $K_{l,m}$. The induced subgraph is the same as $G_{t-1}$, so this cannot form memory modifiable graphs.

\item ($0\rightarrow 0)$$(1\rightarrow 0)$:
$M_{t}$ = $K_{l} + E_{m}$, where $l$ and $m$  are the number of $0$'s and $1$'s respectively. This is in fact a complete split graph consisting of an independent empty set $E_{l}$ to which every vertex of clique $K_{m}$ is completely connected. This is same graph produced by $(0\rightarrow 1)$$(1\rightarrow E)$ as in Proposition~\ref{oneb}.  

\item ($0\rightarrow 1)$$(1\rightarrow \emptyset)$: 
$M_{t}$ = $K_{l,m}$, where $l$ and $m$  are the number of $0$'s and $1$'s respectively. This is same graph produced by $(0\rightarrow1)$$(1\rightarrow0)$ as in Proposition 3.
\item ($0\rightarrow E)$$(1\rightarrow 1)$: $M_{t}$ = $K_{l+m}$, where $l$ and $m$  are the number of $0$'s and $1$'s respectively. This is a complete graph and is also a threshold graph. This is same graph produced by $(0\rightarrow E)$$(1\rightarrow E)$ as in Proposition~\ref{oneb}.

\item ($0\rightarrow E)$$(1\rightarrow \emptyset)$:
 $M_{t}$ = $K_{l} + E_{m}$. This is a complete split graph, which is also threshold graph. This is same graph produced by$(0\rightarrow 1)$$(1\rightarrow E)$ as in Proposition~\ref{oneb}.  

\end{description}

\noindent Hence, any graph $G_{t}$ whose subgraph $G_{t-1}$ is altered with the addition of vertex $t$ is not unique, and does not add further to the number of graphs constructible by $B$. 
\end{proof}

Recall that a graph is $H$\emph{-free} if it does not contain a graph $H$ as
an induced subgraph. Let us consider $(0\rightarrow 1)\left( 1\rightarrow
\emptyset \right) $. The sequence of instructions $%
x_{5}=1_{1}0_{2}0_{3}1_{4}0_{5}$ gives the graph $G_{5}$ on five vertices $%
1_{1},0_{2},0_{3},1_{4},0_{5}$ and edges $\{1_{1},0_{2}\},\{1_{1},0_{3}\},%
\{1_{1},0_{5}\}$, and $\{1_{4},0_{5}\}$. The edges $\{1_{1},0_{3}\},%
\{1_{1},0_{5}\}$, and $\{1_{4},0_{5}\}$ form $P_{4}$. Hence, we can
construct $P_{4}$ in $G_{t}$, meaning that $G_{t}$ is not $P_{4}$%
-free. We can notice, by direct inspection on the graphs in Eq. (\ref{onebg}%
), that every time we attempt to construct $C_{k}$, with $k\geq 5$, we are
forced to include a chord. A \emph{chord} is an edge between two vertices of
a cycle that is not itself an edge of the cycle. Thus, $G_{t}$ is $C_{k}$%
-free for $k\geq 5$. Same situation for $P_{k}$, with $k\geq 5$. The case of 
$P_{5}$ is easy. Consider $P_{4}$ and try to add a vertex -- and an extra
edge -- for obtaining $P_{5}$. In taking $x_{5}=1_{1}0_{2}0_{3}1_{4}0_{5}$,
we can only add $0_{6}$ or $1_{6}$. If we add $0_{6}$, we then create a $4$%
-cycle $1_{1}0_{5}1_{4}0_{6}$. The vertex $1_{6}$ is isolated. There are $32$
binary strings of length $5$. The rule $(0\rightarrow 1)\left( 1\rightarrow
\emptyset \right) $ can not give $P_{5}$. The rules $\left( 0\rightarrow
\emptyset \right) \left( 1\rightarrow \emptyset \right) $ and $(0\rightarrow
1)(1\rightarrow 0)$ are clearly unsuitable for $P_{5}$. The same happens for 
$(0\rightarrow E)\left( 1\rightarrow \emptyset \right) $, $(0\rightarrow
0)\left( 1\rightarrow \emptyset \right) $, since threshold graphs are known
to be $P_{4}$-free. Further, $(0\rightarrow 0)\left( 1\rightarrow 1\right) $
gives only disjoint cliques. With $(0\rightarrow 0)(1\rightarrow 0)$, we
form a clique of $0s$ whenever we add $0$ and so we can not go beyond $P_{3}$%
, which is given by the sequence $0_{1}1_{2}1_{3}$. We always form a new
triangle or a pendant vertex attached to a clique when adding $1$ with the
rule $(0\rightarrow E)(1\rightarrow 0)$. For $(0\rightarrow E)\left(
1\rightarrow 1\right) $, cliques are unavoidable. Finally, $K_{l+m}$ is $%
P_{4}$-free. This discussion leads to a corollary of Proposition \ref{oneb}:

\begin{corollary}
Let \textsf{A} send to \textsf{B} one bit of instructions at each time step $%
t$. Let us assume that \textsf{B} has no randomness. Moreover, let us assume
that \textsf{B} remembers all instructions of \textsf{A}, \emph{i.e.}, one
bit of memory for each vertex. Then, $G_{t}$ is not necessarily $P_{4}$-free
or $C_{4}$-free, but it is $P_{n}$-free and $C_{n}$-free, for each $n\geq 5$.
\end{corollary}

\section{One bit of instructions and one bit of memory for each vertex
(fading memory)}

In this section, we consider the graphs produced when $B$  has `fading' memory. That is, a memory which lasts for an integer \textsf{L} number of time steps. Previously, we considered the cases where \textsf{L} = 1 (no memory) and \textsf{L} = $t$ (perfect memory). Now we consider the graphs formed when \textsf{L} is a small integer. With such a time step, it is impossible to construct a memory modifiable graph.
For the following analysis, we define a linear forest - the disjoint union of path graphs $P_{r_{1}}$, $P_{r_{2}},\ldots,P_{r_{n}}$, each of size $r_{i}$ for $0<i<n$ - to be a graph of the form $\biguplus_{i=1}^{n}P_{r_{i}}$.
Note that an isolated vertex is the path $P_1$, so a linear forest can contain isolated vertices.

We also define the graphs $K'_{(x)}$ and $E'_{(x)}$, with respect to the instruction bit string $x$, as in Section~\ref{sec:oneb}.
The graph $K'_{(x)}$ is $K_{m} + E_{l}$ minus edges of the form $\{t, t'\}$ for $t'< t-1$ when $x_{t}=1$ and $x_{t'}=0$. 
The graph $E'_{x}$ is $\widetilde{K}_{(x)}$ minus edges of the form $\{t, t'\}$ for $t'< t-1$ when $x_{t}=1$ and $x_{t'}=1$.
Furthermore, we use the notation $\oplus$ to represent exclusive `or'.

\begin{proposition}

Let $A$ send to $B$ one bit of instruction at each time step $t$. Furthermore, let \textsf{B} have memory which lasts \textsf{L} = 2 number of time steps. Then, the families of graphs produced will either be a threshold graph, a linear forest, or one of the following graphs: 
 
\begin{equation}
 E'_{(x)} ,\ K'_{(x)} ,\ K_t ,\ E_t.
\label{onebg}
\end{equation}
  
\end{proposition}

\begin{proof} $B$ can distinguish only the labels of vertices $t$ and $t-1$. $B$ cannot distinguish vertices of $G_{t}$ by picking them at random, since $B$ has no randomness.
As before, $m$ and $l$ are the number of $0$'s and $1$'s respectively in the instruction string $x$.    
The interpretations of a one bit instruction from \textsf{A} are the same as in Proposition~\ref{oneb}. Given these constraints, we proceed with a case by case analysis. 

\begin{description}
\item[$\left( 0\rightarrow\emptyset\right) \left( 1\rightarrow
\emptyset\right) $:]  $G_{t}$ = $E_{t}$. 
\item[$(0\rightarrow E)\left( 1\rightarrow E\right) $:] 
 $G_{t}$ = $K_{t}$
\item[$(0\rightarrow E)\left( 1\rightarrow\emptyset\right) $:]
 $G_{t}$ = $\mathcal{T}$, again, as this is the case of Proposition~\ref{prop2}. 
\item[$(0\rightarrow0)\left( 1\rightarrow\emptyset\right) $:]
$G_{t}$ = $\biguplus_{r \in R(x) } P_{r}$ for $R: \{0, 1\}^* \to \mathbb{N}^*$.
The string $R(x)$ is the ordered sequence counting consecutive zeros that appear in a bit string $x$, for instance $R(00110100010) = (2,1,3,1)$.

\item[$(0\rightarrow0)\left( 1\rightarrow1\right) $:] 
$G_{t}$ = $\biguplus_{r \in R(x) } P_{r} \uplus \biguplus_{s \in S(x) } P_{s}$, for $R: \{0, 1\}^* \to \mathbb{N}^*$ and $S: \{0, 1\}^* \to \mathbb{N}^*$.
The strings $R(x)$ and $S(x)$ are the ordered sequences counting consecutive zeros and ones respectively that appear in a bit string $x$, \emph{e.g.} $R(00110100010) = (2,1,3,1)$, $S(00110100010) = (2,1,1)$.

\item[$(0\rightarrow1)\left( 1\rightarrow\emptyset\right) $:] 
$G_{t}$ = $\biguplus^r_{i=1} P_2 \uplus E_{t-2r}$, where $r$ is the number of occurrences of the substring $(10)$ in the instruction string $x$. 

\item[$(0\rightarrow1)(1\rightarrow0)$:] 
$G_{t}$ = $\biguplus_{q \in Q(x)} P_q \uplus E_{t - \sum_{q\in Q(x)} q}$, for  
$Q: \{0, 1\}^* \to \mathbb{N}^*$.
The string $Q(x)$ is the ordered sequence of the lengths of substrings of alternating bits appearing in a bit string $x$, for instance $Q(00110100010) = (2,4,3)$, counting bits $(x_2x_3)$, $(x_4x_5x_6x_7)$ and $(x_9x_{10}x_{11})$.

\item[$(0\rightarrow0)(1\rightarrow0)$:] 
$G_{t}$ = $\biguplus_{a\in A(x)} P_{a} \uplus E_{t - \sum_{a\in A(x)} a}$, for $A: \{0, 1\}^* \to \mathbb{N}^*$.
The string $A(x)$ is the ordered sequence of the lengths of substrings of the form $(0 \cdots 01)$ appearing in a bit string $x$, for instance $Q(00110100010) = (3,2,4)$, counting bits $(x_1x_2x_3)$, $(x_5x_6)$ and $(x_7x_8x_9x_{10})$.

\item[$(0\rightarrow E)(1\rightarrow0)$:]
$G_{t}$ = $K'_{(x)}$.
We see this as the edge $\{i,j\}$ exists in $G_t$ if and only if: $\ell(i) = 0$, $i > j$, or $\ell(i) = 0$, $\ell(j) = 1$, $i+1 <j$, which precisely agrees with the definition of $K'_{(x)}$.

\item[$(0\rightarrow E)\left( 1\rightarrow1\right) $:] 
$G_{t}$ = $E'_{(x)}$.
We see this as the edge $\{i,j\}$ exists in $G_t$ if and only if: $\ell(i) = 0$, $i > j$, or $\ell(i) = \ell(j) = 1$, $i+1 = j$, which precisely agrees with the definition of $E'_{(x)}$.

\end{description}
\end{proof}

\section{Randomness only}\label{sec:randomness}

Randomness is an indispensable resource in communication and computational
complexity, and in general is a fundamental tool in many areas of computer
science, statistical mechanics, \emph{etc.} How can we use randomness to
construct $G$?\ Start with the empty graph on $n$ vertices and insert each
edge with probability $1/2$.
The outcome of this random process is the (uniform) random graph and it is denoted by $G_{n,1/2}$ (see \cite{J90}).
The probability that $G_{n,1/2}\cong H$ is nonzero, but exponentially small; in fact, $2^{-\binom{n}{2}}$ if we do not keep into account the size of the isomorphism class.
It follows however that without instructions we can still construct $G_{n,1/2}$.
The randomness used amounts to flip $C(n,2)$ coins, which corresponds to an equal number of random bits. 
The time required for the process consists of a single time-step, since we can
flip all the coins at once.
So, bits here quantify both randomness and information. A quick observation:\ when the length of the instructions is zero, we can construct $G_{n,1/2}$ (with nonzero probability) by the use of $C(n,2)$ bits assigned to the random process; conversely, without randomness, we need $C(n,2)$ bits of information to construct $G$ (with probability one).
An important fact is that the random process is instantaneous.
Randomness is generated by unbiased coins that can be re-used as many times we need. 
On the other hand, if we include time, each random bit can be given by flipping
the same coin -- to be discussed further.

Notice that even if we do not use any instructions, we still need to be able
to identify the dyads of vertices when flipping the coins.
Suppose we have $n=4$. We have $C(4,2)=6$ dyads.
When we flip the coin at time $t=1$, we need to choose two vertices.
Without such a choice, we will not be able to add an edge.
More precisely, we need labels on the vertices so that we can remember whether we have already flipped a coin for a given dyad. Suppose we flip a
coin for $(1,2)$ at time $t=1$.
At time $t=2$, we need to choose another dyad, say $(1,3)$.
The dyad can be chosen only if we remember that we have
already dealt with $(1,2)$.
There are two immediate options: either we have a list of vertex labels, which is then information; or the dyads need to be
chosen at random at each time step. In the latter case we definitely need
more randomness, in terms of more coin flips, but we still can construct the
graph. In fact nothing really happens if a dyad is chosen twice. What if we
are allowed more than a single time-step?\ \emph{Time} is in principle
another plausible resource. We consider a topological version of time, where
each time-step is simply a positive integer, $t\in\mathbb{Z}^{\geq1}$. Time
opens the way to construct $G$ not all at once, but via a sequence of
operations. With time available we can distribute instructions over a number
of different time steps. Hence the amount of information needed at each time
step could be dosed. This is the direction that we will take. Therefore we
choose to look at time not as a resource but as a parameter which is
unmodifiable by resources. The clock will tick whenever we perform an action (%
\emph{e.g.}, adding a vertex).

If we use randomness, time could be subdivided into three operations:\
adding vertices; flipping coins, adding edges. Of course, $G_{n,1/2}$ can be
constructed by flipping $C(n,2)$ coins at the same time, or by flipping a
single coin $C(n,2)$ times. In the second case, if we allow a single coin as
a physical resource to generate randomness, the time for constructing the
graph is $t=C(n,2)$. As might be expected it does not make much sense to
talk about random random graphs, where the pairs of vertices are chosen at
random. In this case, a dyad has probability $1/n^{2}$ to be picked. If we
wait long enough then we just obtain the random graph.

At time $t=1$, we have a graph $G_{1}$, at time $t=2$, $G_{2}$, and finally $%
G_{n}$, at time $t=n$. Notice that not every time-step $t$ needs to
correspond to a graph on $t$ vertices. However, we may assume without loss
of generality that a graph $G_{t^{\prime}}$ has a number of vertices $%
k^{\prime }\geq k$, whenever $t^{\prime}>t$ and $G_{t}$ is on $k$ vertices.
In this notation $G_{n}=G$ at the end of the growth process. For the
moment, let us consider a basic case: we start with $G_{1}=K_{1}$, \emph{i.e.%
}, the empty graph with a single vertex, and add a new vertex at each
time-step. Moreover, suppose that our only available resource is randomness.
As a consequence of this fact, at each time step $t$ we add a single vertex
and choose its neighbours at random. We end up with the following process
introduced in \cite{JS13} and originally studied in the context of graph
limits (see also \cite{L12}). For each $t$, let $\nu_{t}$ be a probability
distribution on $\{1,2,\ldots,t\}$. By denoting as $D_{i}$ a random variable
drawn according to the distribution $\nu_{i}$, we obtain $G_{i}$ by adding a new vertex to $%
G_{i-1}$ and connecting it to a subset of size $D_{i}$ in $V(G_{i-1})$
distributed with respect to $\nu_{i}$. When $\nu_{i}=$ Bi$(i-1,p)$, with $%
p\in\lbrack0,1]$ and $i\geq2$, we get the Erd\H{o}s--R\'{e}nyi random graph.
By modifying $\nu_{i}$, we can end up with various other graph ensembles.

The above process needs randomness and no information at all. How much
randomness?\ 

\begin{proposition}
Let $G_{n}$ be constructed only with the use of randomness and no
information. Then, the amount of randomness needed to construct $G_{n,1/2}$\
and $G_{n}$\ are asymptotically equivalent.
\end{proposition}

\begin{proof}
We quantify randomness by the number of random bits needed to perform each
choice. The table below lists these bits for the graphs $G_{3},G_{4}$, and $%
G_{5}$. Notation:\ $y_{i}^{j}$ is the $j$-th bit used for choosing
neighbours of vertex $i$ in $G_{i-1}$ and $x_{i}^{j}$ is the $j$-th bit used
for choosing their number, both at time $i$. For $G_{3},G_{4}$, and $G_{5}$,
we need $4,8,14$ bits, respectively:%
\begin{equation*}
\begin{tabular}{ll}
$G_{3}\text{:\quad }$ & $\left( x_{1}^{1}\right) \left(
(x_{3}^{1}x_{3}^{2})(y_{3}^{1})\right) ;$ \\ 
$G_{4}\text{:\quad }$ & $(G_{3})\left(
(x_{4}^{1}x_{4}^{2})(y_{4}^{1}y_{4}^{2})\right) ;$ \\ 
$G_{5}\text{:\quad }$ & $(G_{4})\left(
(x_{5}^{1}x_{5}^{2}x_{5}^{3})(y_{5}^{1}y_{5}^{2}y_{5}^{3})\right) .$%
\end{tabular}%
\ \ \ \ \ \ \ \ 
\end{equation*}%
We have denoted by $(G_{i})$ the random bits for $G_{i}$. The formula for
this integer sequence is $a(1)=0$, $a(2)=1$, \ and for $n\geq 3$, we have (proof below)%
\begin{equation}
a(n)=a(n-1)+b_{o}(n)+\left\lfloor \log _{2}(n-1)\right\rfloor +1,
\label{aodd}
\end{equation}%
for $n$ odd and 
\begin{equation}
a(n)=a(n-1)+b_{e}(n)+\left\lfloor \log _{2}(n-1)\right\rfloor +1,
\label{aeven}
\end{equation}%
for $n$ even, where%
\begin{equation*}
\begin{tabular}{lll}
$b_{e}(n)=\left\lfloor \log _{2}\left( \binom{n-1}{n/2}-1\right)
\right\rfloor +1$ & and & $b_{o}(n)=\left\lfloor \log _{2}\left( \binom{n-1}{%
(n-1)/2}-1\right) \right\rfloor +1.$%
\end{tabular}%
\ \ 
\end{equation*}%
The integer $\left\lfloor \log _{2}(n)\right\rfloor +1$ is the number of
bits in the binary expansion of $n$. For Eqs. \ref{aodd} and \ref{aeven},
let us consider the growth process. At time $t=1$, $G_{1}=K_{1}$. At time $%
t=2$, we add a vertex $2$. There are $2$ possible cases: we flip a coin and
get either $d_{G_{2}}(2)=0$ or $d_{G_{2}}(2)=1$. (Recall that the degree $%
d_{G}(i)$ is the number of neighbours of a vertex $i\in V(G)$.) Hence, $%
a(2)=1$. At time time $t=3$, we add a vertex $3$. There are $3$ possible
cases: $d_{G_{3}}(3)=0,\ldots,d_{G_{3}}(3)=2$. It is evident that the
contribution from this term at time $t=n$ is then $\left\lfloor \log
_{2}(n-1)\right\rfloor +1$ because $d_{G_{n}}(n)=0,\ldots,d_{G_{n}}(n)=n-1$.
Vertex $3$ can choose among $2$ vertices, and in general vertex $n$ can
choose among $n-1$ vertices. What is the $d_{G_{n}}(n)$ with the highest
randomness cost?\ Let us label as $1,2\ldots,n-1$, the $n-1$ vertices in $G_{n}$
potentially adjacent to vertex $n$. The answer is $d_{G_{n}}(n-1)=n/2$ for $%
n $ even and $d_{G_{n}}(n-1)=(n-1)/2$ for $n$ odd. These are the numbers
giving the largest binomial coefficient. Since the binary system starts
enumerating from $0$, we take the number of bits in $C(n-1,n/2)-1$ and $%
C(n-1,(n-1)/2)-1$. Summing up everything, a formula for $n\geq 4$ is%
\begin{eqnarray*}
a(n) &=&\sum_{i=3\text{; even}}^{n-1}\left\lfloor \log _{2}\left( \binom{i-1%
}{i/2}-1\right) \right\rfloor +\sum_{i=3\text{; odd}}^{n-1}\left\lfloor \log
_{2}\left( \binom{i-1}{(i-1)/2}-1\right) \right\rfloor \\
&&+\sum_{i=2}^{n}\left\lfloor \log _{2}(i-1)\right\rfloor +2n-3.
\end{eqnarray*}%
Let us look at the asymptotic behaviour of $a(n)$. First, $\binom{i-1}{i/2}%
\leq 2^{i-1}$ and $\binom{i-1}{(i-1)/2}\leq 2^{i-1}$. Also, there is a
constant $c>0$ such that $n!\leq cn^{n+1/2}e^{-n}$ when $n\rightarrow \infty 
$. By combining these facts, one can see that the asymptotic efficiency
class of $a(n)$ is $\Theta (n^{2})$.
Notably the same happens for $C(n,2)$, which is the number of random bits needed for the uniform random graph, $G_{n, 1/2}$.
\end{proof}

\bigskip

Let consider again the process above, but where the probability distribution
is taken to be uniform -- see \cite{BMS13}. We iteratively construct a graph 
$G_{t}=(V,E)$, starting from $G_{1}=K_{1}$. The $t$-th step of the iteration
is divided into three substeps: \emph{(1)}\ We select a number $k\in
\{0,1,\ldots,t-1\}$ with equal probability. Assume that we have selected $k$. 
\emph{(2)}\ We select $k$ vertices of $G_{t-1}$ with equal probability.
Assume that we have selected the vertices $v_{1},v_{2},\ldots,v_{k}\in
V(G_{t-1})$. \emph{(3)}\ We add a new vertex $t$ to $G_{t-1}$ and the edges $%
\{v_{1},t\},\{v_{2},t\},\ldots,\{v_{k},t\}\in E(G_{t})$. For a graph $G$ on $t$
vertices, the\ \emph{likelihood} of $G$, denoted by $\mathcal{L}(G)$, is
defined as the probability that $G_{t}=G$, where $G_{t}$ is the graph given
by the above iteration: $\mathcal{L}(G):=\emph{Pr}[G_{t}=G]$. For example, $%
\mathcal{L}(K_{t})=1/t!$ and $\mathcal{L}(K_{1,t-1})=\frac{t}{(t!)^{2}}%
\sum_{i=0}^{t-1}i!$, where $K_{1,n-1}$ is a star on $n$ vertices. An
important point is a link between the likelihood and the size of the
automorphism group of $G$. An \emph{automorphism} of a graph $G=(V,E)$ is a
permutation $p:V(G)\longrightarrow V(G)$ such that $\{v_{i},v_{j}\}\in E(G)$
if and only if $\{p(v_{i}),p(v_{j})\}\in E(G)$. The set of all automorphisms
of $G$, with the operation of composition of permutations \textquotedblleft $%
\circ $\textquotedblright , is a subgroup of the permutation group denoted by Aut$(G)$. It
is in fact possible to show that 
\begin{equation*}
\frac{1}{|\operatorname{Aut}(G)|\prod_{i=1}^{t}\binom{i-1}{\lfloor
(i-1)/2\rfloor }}\leq \mathcal{L}(G)\leq \frac{1}{|\operatorname{Aut}(G)|}.
\end{equation*}%
It is plausible to conjecture that the minimum likelihood is attained by the
complete bipartite graph on $n$ vertices, $K_{p-1,p}$, when $n=2p-1$, and $%
K_{p,p}$, when $n=2p$. Numerical evidence is exhibited in \cite{W}. These
complete bipartite graph have a relatively large automorphism group and by
Mantel's theorem are the triangle-free graphs with the largest number of
edges. The analogue conjecture for the maximum seems harder to state. The
computational complexity of the likelihood is an open problem. The original
motivation for introducing the likelihood was to measure how likely is that
a given graph is generated at random. The idea fits the context of
quasi-randomness, \emph{i.e.}, the study of how much a given graph resembles
a random one.

\section{Trees built from limited resources}

We can consider also building trees using the framework developed thus far. 
A \emph{tree} is a graph without cycles.
We grow trees by adding vertices one-by-one. Given a tree $T_{n}$ on $n$ vertices, we can always identify a
vertex $1$, called the \emph{root}, and added at time $t=1$. The $k$\emph{%
-th generation} of the tree are the vertices at distance $k$ from the root.
The \emph{leaves} are pendant vertices, \emph{i.e.} vertices of degree $1$.
The vertex $n$ is always a leaf. Also the root, in our definition, can be a
leaf. When we grow a tree by adding vertices one-by-one then we also add
edges one-by-one. In fact, the number of edges of a tree $T_{n}$ is exactly $n-1$.

\subsection{Trees built with random bits}

The process in Section~\ref{sec:randomness} can be used to grow trees if at each
time step it is restricted to add a single edge. Equivalently, the degree of
vertex $i$ in $T_{i}$ is $1$. It can be done by bypassing the random choice
of the degree. If we keep the random choice of the adjacent vertex, we have
a model of a random tree. At each time step $i$ we add a new vertex and a new
edge incident with that vertex. The neighbour of the vertex is chosen
randomly in $T_{i-1}$. This process is also called the \emph{uniform attachment
model} and $T_{i}$ is usually denoted by UA$(i)$ \cite{SM95}. The graph UA$%
(i) $ can be also obtained by taking one of all the possible spanning trees
of $K_{n}$ at random; the well-known Pr\"{u}fer bijection guarantees that $%
K_{n}$ contains all trees on $n$ vertices. The process gives a sequence of
random recursive trees, where a tree on the vertices $\{1,2,..,n\}$ is \emph{%
recursive} if the vertex labels along the unique path form $1$ to $j$
increase for every $j\in \{2,3,\ldots,n\}$.

For a rooted tree $T$, let $L(T)$ denote the set of its leaf vertices.
We also denote by $P_T(l)$ the (unique) path from the root vertex to the leaf $l \in L(T)$. 
For graphs $G$ and $H$ we denote by $G \cup H$ the \emph{graph union} of $G$ and $H$.
This is the graph $(V(G) \cup V(H), E(G) \cup E(H))$.
Also, let $b(n):= \lfloor \log_2 n \rfloor + 1$, the number of bits needed to represent the integer $n$ as a binary string.

\begin{proposition}
    Let $T$ be any tree on $n$ vertices and let $T_n = \operatorname{UA}(n)$ be the random tree, as constructed previously.
    Then, $\mathcal{L}(T) > 0$ where $\mathcal{L}$ is the likelihood function defined in Section~\ref{sec:randomness}, that is, $\mathcal{L}(T) = \operatorname{Pr}(T = T_n)$.
\end{proposition}
\begin{proof}
    We can describe the tree $T$ as the graph union of the paths to each of its leaves, that is, $T = \bigcup_{l \in L(T)} P_T(l)$.
    Each path $P_T(l)$ takes the form $(1, i_1, i_2, \ldots, l)$, so any individual given path can be generated according to our random process, as there is always nonzero probability for the edge $\{i,j\}$ for $j > i$.
    
    We now need to prove that our random process supports constructing the union of all paths in $T$.
    Choose any two leaves in $T$, $l_1, l_2 \in L(T)$.
    The induced paths from the root of $T$ to $l_1$ and $l_2$ respectively are $P_T(l_1) = (1, i_1, i_2, \ldots, l_1)$ and $P_T(l_2) = (1, j_1, j_2, \ldots, l_2)$.
    From the previous paragraph, there is nonzero probability that $P_T(l_1)$ and $P_T(l_2)$ are individually subgraphs of $T_n$.
    For these paths to have zero probability of simultaneously being in $T_n$, there must be a vertex $i_r = j_s$ for some $r,s$ such that there are edges $\{ i_r, i_{r'}\}, \{ j_s, j_{s'} \}$ for $i_r < i_r'$, $j_s < j_s'$ and $i_{r'} \neq j_{s'}$.
    This is because every vertex $t \in \{1, 2, \ldots, n\}$ has only one neighbour smaller than itself $t' < t$.
    Suppose the condition holds.
    Then, the vertices $\{ 1, i_1, \ldots, i_{r'}, \ldots, i_r, j_1, \ldots , j_{s'}, \ldots, j_s \}$ induce a cycle.
    However, $T$ is a tree, so has no induced cycle and we have a contradiction.
    This means $P_T(l_1) \cup P_T(l_2)$ is a subgraph of $T_n$ with nonzero probability. 
    
    Since the argument holds for any $l_1, l_2 \in L(T)$, we have that there is nonzero probability that $T = \bigcup_{l \in L(T)} P_T(l)$ is a subgraph of $T_n$.
    Indeed, since both $T$ and $T_n$ are trees on $n$ vertices, they both have $n-1$ edges and so $T\subseteq T_n $ implies that $ T=T_n$ and the result follows.
\end{proof}

\subsection{Trees built with instructions and memory}

We can also consider the model earlier in the text, where Alice (\textsf{A}) sends instructions to Bob (\textsf{B}), and \textsf{B} has memory.
We can bound the size of the instructions and memory needed to construct an arbitrary tree $T$.

\begin{proposition}\label{prop7}
Let $T_n$ be a tree on $n$\ vertices. Then $T_n$\ can be constructed
with the use of $O(n\log n)$ bits of instructions and $O(n\log n)$\ bits of memory.
\end{proposition}
\begin{proof}
    Consider constructing tree $T_t$ given that we have the tree $T_{t-1}$.
    Since $T_{t}$ has $t-1$ edges and $T_{t-1}$ has $t$ edges, we only add one edge, between vertex $t$ and some other vertex $t' < t$.
    The instructions have to specify which vertex $t'$ will be $t$'s neighbour.
    This requires $b(t-1)$ bits.
    We also require a label in memory for each vertex $t'$, each of which requires $b(t')$ bits.
    Summing up, for the graph $T_n$, we require $\sum^{n-1}_{t=1}b(t)$ bits of memory for the vertex labels and $\sum^{n-1}_{t=1}b(t)$ bits for the instructions.
    Finally,
    \begin{equation*}
        \begin{aligned}
            \sum^{n-1}_{t=1}b(t) &=  \sum^{n-1}_{t=1} \left(\lfloor \log_2 n \rfloor +1 \right) \leq  \sum^{n-1}_{t=1} \log_2 n  + n-1 = \log_2((n-1)!) + n-1 \\
            &\leq (n-1)\left( \log_2(n-1) - \log_2(\mathrm{e}) + 1 \right) +O(\log_2(n-1)) = O(n \log n).
        \end{aligned}    
    \end{equation*}    
\end{proof}

\section{Conclusions}

In this work we have considered a scenario where two parties Alice (\textsf{A}) and Bob (\textsf{B}) construct a graph together, given limited communication, memory and randomness.
We enumerate the different classes of graphs that \textsf{A} and \textsf{B} can construct under different constraints on these resources.

A number of open questions remain: for instance, in Section~\ref{sec:oneb} we saw that giving Bob one bit of memory per vertex lifted the class of threshold graphs to a larger class.
Indeed, it is well known that computational problems such as graph isomorphism and maximum clique are easy when the instances are taken as threshold graphs.
Perhaps for this extended class of graphs such problems are also easy.
Indeed, it would be interesting to construct a hierarchy of graphs with an increasing number of bits of memory and at which point these problems increase in difficulty.
Maybe there there is some correspondence with some well-known class of graphs.

The case for fading memory where $2 < \mathsf{L} < t$ needs to be investigated in a similar fashion as does the case of $\mathsf{B}$ having access to differing numbers of random bits.

We may also ask the converse question: given a graph $G$, how many bits do $\mathsf{A}$ and $\mathsf{B}$ need to build it?
Recognising threshold graphs is linear-time~\cite{CH73,HIS78}.
Is there a polynomial-time procedure for determining the number of bits of instructions and memory needed to construct $G$, beyond the $\Theta(n^2)$ bound? 

\bigskip

\noindent \emph{Acknowledgments. }This paper follows a talk given by SS at the
\textquotedblleft Doctoral Workshop on Mathematical and Engineering Methods
in Computer Science (MEMICS)\textquotedblright , Tel\v{c}, Czech Republic,
October 23--25, 2015. SS would like to thank the organisers of MEMICS for
their kind invitation. SS would also like to thank Jan Bouda for being a
great host. 
This paper forms part of a Year 2 undergraduate project for AM.
We\ also would like to thank Nate Ackerman, Cameron Freer, Christopher Banerji, Josh Lockhart,
Octavio Zapata, and Vern Paulsen, for helpful discussion and their patience.

\end{document}